\pdfoutput=1 

\documentclass[11pt]{article}
\usepackage[utf8]{inputenc}
\usepackage[T1]{fontenc}
\usepackage{graphicx}
\usepackage{epstopdf}
\usepackage{hyperref}
\usepackage{latexsym}
\usepackage{amsmath}
\usepackage{amssymb}
\usepackage{amsthm}
\usepackage{dsfont}
\usepackage[capitalise]{cleveref}
\usepackage{subcaption}

\usepackage{natbib}
\bibpunct[, ]{(}{)}{;}{a}{,}{,}

\usepackage[top=3cm,bottom=2cm,left=3cm,right=3cm,marginparwidth=1.75cm]{geometry}

 \newcommand{\bq}{\begin{equation}}
 \newcommand{\eq}{\end{equation}}
 \newcommand{\bqn}{\begin{eqnarray}}
 \newcommand{\eqn}{\end{eqnarray}}

\renewcommand{\P}{ \mathbb P }
\newcommand{\EXP}{\mathbb{E}}

\newcommand{\ud}[1]{\, \mathrm{d}#1}
\newcommand{\deriv}[3][]{\frac{\ud^{#1} \hspace{-0.3mm} #2}{\ud{#3}^{#1}}}

\newcommand{\R}{\mathbb{R}}

\usepackage{graphicx}
\usepackage[colorinlistoftodos]{todonotes}

 \newtheorem{Thm}{Theorem}

\newtheorem{Prop}[Thm]{Proposition}

\title{The Stochastic Gause predator-prey model: noise-induced extinctions and invariance}
\author{
  \small Leon Alexander Valencia, Ph.D\\
  \small Jorge Mario Ramirez osorio, Ph.D\\
  \small Jorge Andrés Sánchez Arteaga, Ph.D\\
  \date{}
}
\begin{document}
\maketitle

\begin{abstract}
This paper explores a stochastic Gause predator-prey model with bounded or sub-linear functional response. The model, described by a system of stochastic differential equations, captures the influence of stochastic fluctuations on predator-prey dynamics, with particular focus on the stability, extinction, and persistence of populations. We provide sufficient conditions for the existence and boundedness of solutions, analyze noise-induced extinction events, and investigate the existence of unique stationary distributions for the case of Holing Type I functional response. Our analysis highlights the critical role of noise in determining long-term ecological outcomes, demonstrating that even in cases where deterministic models predict stable coexistence, stochastic noise can drive populations to extinction or alter the system's dynamics significantly. 

\end{abstract}

\section{Introduction}

The study of predator-prey dynamics has been an area of interest in mathematical ecology for decades. These models provide insight into species interactions in various environments and have been widely used to explore phenomena such as population coexistence and extinction. In particular, stochastic models have gained importance in recent years due to their ability to capture the inherent uncertainty in ecological and environmental processes \citep{mao2002environmental, abbott2017alternative}.

A classical approach to modeling these interactions is through deterministic ordinary differential equations (ODEs), such as those presented by Lotka and Volterra. However, in real-world scenarios, random fluctuations play a fundamental role, leading to the development of stochastic models that incorporate environmental and demographic noise. Khasminskii's work \cite{khasminskii2012stability} provides a solid foundation for stochastic stability in dynamic systems, while other authors like Kot \cite{kot2001elements} and Murray \cite{murray2002mathematical} have significantly contributed to the theory of ecological models.

Within this context, stochastic differential equations (SDEs) have proven to be useful tools for describing the behavior of populations in uncertain environments. For example, studies on the stability of population systems in the presence of environmental noise have been fundamental in understanding how noise can prevent or cause species extinction \cite{mao2011stationary, zhou2020persistence}. The incorporation of Brownian-type noise in population models has revealed that stochastic fluctuations can stabilize or destabilize systems, depending on initial conditions and environmental characteristics \cite{liu2016dynamics, liu2011survival}.

In this article, we study a stochastic 
 Gause-type predator-prey model describing the evolution in time $t \geq 0$ of the populations of prey $x(t)$ and predator $y(t)$ via the  following system of stochastic differential equations
\begin{equation}\label{eq_DimStochGause}
\begin{split}
\ud x & = \left(rx\left(1-\dfrac{x}{K}\right)- f(x)y\right)\ud t+ x s_x \ud B_{x}, \\
\ud y & = \left(bf(x)y-m y\right)\ud t+y s_{y} \ud B_{y},
\end{split}
\end{equation}
where $f$ is a non-negative function called the \textit{functional response} of the predator to the prey. It describes the dependence between the predator's behavior and the availability of prey \cite{mao2007stochastic,li2021permanence,zhou2020persistence}, an plays a crucial role in the dynamics of the system, both in the deterministic and stochastic models. Here, we consider functional responses of various types, the main ones being of Holling types I to IV as listed in \cref{tab_FRadim}. 

The model parameters in \cref{eq_DimStochGause} are all assumed positive and interpreted as follows: $r$ is the intrinsic growth rate of the prey, $K$ its carriying capacity, $b$ is the interaction coefficient of the predator, and $m$ its mortality rate. See \cite{volterra1926fluctuations,murray2002continuous,brauer2012mathematical}. The diffusion terms in \cref{eq_DimStochGause} are driven by standard, independent Wiener processes $B_x, B_y$ and have diffusion coefficients that depend linearly on $x$ and $y$ with intensities $s_x,s_y>0$. This additive-noise terms model the situation in which, for example, there is time-dependent uncertainty in the parameters $r$ and $m$. See \cite{mao2002environmental}

One of the key challenges in stochastic predator-prey models is determining the conditions under which populations can persist in the long term. This type of analysis has been approached from multiple perspectives, including Markov process theory \cite{meyn1993stability, bharucha1997elements}. For instance, Foster-Lyapunov criteria have been used to evaluate the stability of continuous processes and to determine extinction probabilities in stochastic ecological systems \cite{meyn1993stability}. The research by Liu and collaborators \cite{liu2016dynamics} has shown how Holling type II and type III functional response models affect the stability and permanence of populations in fluctuating ecosystems.

In addition, recent studies have explored the long-term behavior of predator-prey systems under stochastic influences. The work of Li and Guo \cite{li2021permanence} focuses on the permanence of stochastic models with general functional responses, while Xu et al. \cite{xu2020analysis} investigate more complex systems with modifications to Leslie-Gower and Holling type IV schemes. These investigations have demonstrated that incorporating more complex behaviors, such as group defense and nonlinear responses to prey population growth, can significantly impact system dynamics \cite{pei2022extinction}.

We provide sufficient conditions for existence and boundedness of solutions to \cref{eq_DimStochGause}, extinction of prey and/or predator, and existence of a unique stationary distribution. Except for our result on invariant distributions which holds only for Type I Holling functional responses, our results do not assume a specific form of $f$. Rather, we always assume that $f$ in \cref{eq_DimStochGause} is non-negative, satisfies $f(0) = 0$ and is either bounded or grows sub-linearly. Our techniques rely mostly on the theory 
of stability for stochastic differential equations,  specifically on applications of the comparison theorem and construction of appropriate Lyapunov methods; the book of \cite{khasminskii2012stability} being a good source for this material.

The organization of the this paper is as follows. In \cref{sec:prel} we review the structural properties of the deterministic Gause model and introduce the stochastic model in non-dimensional variables. Subsequently, \cref{sec:prel} contains preliminary results that ensure the well-posedeness of the stochastic problem, which in this context includes existence, uniqueness, positivity and boundedness of the solutions. The analysis then explores noise-induced extinction phenomena, demonstrating that certain noise intensities can lead to the eventual extinction of one or both species, even in situations where the deterministic model predicts coexistence. Finally, we provide conditions for the existence of an invariant distribution in the case of a type I functional response, ensuring the persistence of populations under certain parameter configurations of the model.

The results obtained underscore the significant impact that uncertainty, modeled as stochastic noise, can have on the dynamics of ecological systems, altering behaviors predicted by deterministic models and leading to extinction or coexistence scenarios that depend on the level of noise





\section{Deterministic and stochastic Gause models}

We begin with a dimensionless formulation of the deterministic Gause model and a review of its relevant structural properties. Namely, the following system of ordinary differential equations
\begin{equation}\label{eq_DimGause}
\frac{\ud x}{\ud t}  = rx\left(1-\dfrac{x}{K}\right)- f(x)y,\quad
\frac{\ud y}{\ud t}  =bf(x)y-m y,
\end{equation}
where all model parameters are assume positive. We distinguish between two different classes of the functional response $f$ in \cref{eq_DimGause}: linear response (also known as of Holling Type I); and the class of bounded, non-negative continuous functional responses such that $f(0) =0$, which includes Holling's Types II, III and IV. See \cref{tab_FRadim}. For both classes we arrive at the same non-dimensional model but use slightly different normalizations.

For a linear functional response $f(x) = cx$ for some $c>0$, we change to dimensionless variables by
\begin{equation}\label{covDetLin}
    \frac{x}{K} \to x, \quad \frac{c}{r}y \to y, \quad r t \to t, 
\end{equation}
and define the following dimensionless functional response and parameters 
\begin{equation}\label{def_phibd1}
    \varphi(x) = x, \quad \beta = \frac{bcK}{r}, \quad \delta = \frac{m}{r}.
\end{equation}
For a bounded functional response $f(x)$ with $F:=\sup_{x\geq 0} f(x) < \infty$, we change variables by
\begin{equation}\label{covDetBded}
    \frac{x}{K} \to x, \quad \frac{F}{Kr}y \to y, \quad r t \to t, 
\end{equation}
and define 
\begin{equation}\label{def_phibd2}
    \varphi(x) = \frac{f(K x)}{F}, \quad \beta = \frac{Fb}{r}, \quad \delta = \frac{m}{r}.
\end{equation}
Under the corresponding change of variables and notation, the Gause model \cref{eq_DimGause} takes the adimensional form
\begin{equation}\label{eq_DetGause}
    \deriv{x}{t} = x(1-x) - \varphi(x) y, \quad \deriv{y}{t} = \beta \varphi(x) y -\delta y.
\end{equation}

Since $\varphi(0) = 0$, \cref{eq_DetGause} always has two equilibrium points corresponding to extinction events: the global extinction state $(0,0)$ and predator extinction $(1,0)$, furthermore,
\begin{equation}\label{cond_Determ}
    (0,0) \text{ is always unstable, } \quad (1,0) \text{ is asymptotically stable if and only if } \beta < \beta^*:=\frac{\delta}{\varphi(1)}.
\end{equation}
A bifurcation occurs at $\beta = \beta^*$, giving rise to various types of stable manifolds that guarantee persistence of both species, including stable solutions and limit cycles. These have been characterized for some functional response types and parameter ranges in \cite{kot2001elements}.

The corresponding dimensionless stochastic Gause model can be obtained by performing the change of variables \eqref{covDetLin} or \eqref{covDetBded} in \cref{eq_DimStochGause}, along with the normalization of the noise intensities
\begin{equation}
    \sigma_x = \frac{s_x}{\sqrt{r}}, \quad \sigma_y = \frac{s_y}{\sqrt{r}}
\end{equation}
We thus arrive to our system of interest,
\begin{equation}\label{eq_StochGause}
	\ud x= (x(1-x) -\varphi(x)y) \ud t + x \sigma_x \ud B_x, \quad
	\ud y =  (\beta \varphi(x) y- \delta y) \ud t + y \sigma_y \ud B_y.
\end{equation}
where the parameters $\beta,\delta$ are as in \cref{def_phibd1,def_phibd2}, $B_x$ and $B_y$ are independent Brownian processes and $\varphi$ is a continuous, positive function such that $\varphi(0) = 0$. Note that under this hypotheses, the coefficients of \cref{eq_StochGause} are locally Lipschitz, so standart SDE theory guarantees the existence of strong solutions.

\begin{table}[ht]
\centering
\begin{tabular}{|c|c|c|c|c|}
\hline
Type & $f(x)$ & $F$ & $\varphi(x)$ & parameters \\
\hline
I & $cx$ & $\infty$ & $x$ &  \\
\hline
II & $\frac{c x}{x+a}$ & $c$ & $\frac{x^2}{x + \alpha}$ & $\alpha = \frac{a}{K}$ \\
\hline
III & $\frac{c x^2}{x^2+a}$ & $c$ & $\frac{x^2}{x^2 + \alpha}$ & $\alpha = \frac{a}{K^2}$ \\
\hline
IV & $\frac{c x}{\frac{x^2}{i}+x+a}$ & $\frac{ci}{1+2 \sqrt{ai}}$ & $\frac{(\alpha + 2 \sqrt{\gamma})x}{\gamma + \alpha x + x^2}$ & $\alpha = \frac{i}{K}, \, \gamma = \frac{ai}{K^2}$ \\
\hline
\end{tabular}
\caption{Functional responses, their adimensionalization and parameters}
\label{tab_FRadim}
\end{table}

\subsection{Regularity and boundedness}\label{sec:prel}

In this section we present fundamental results concerning the model in \cref{eq_StochGause}.  The first result  guarantees that the solution to \cref{eq_StochGause} for functional responses of all types in \cref{tab_FRadim} is well-posed and `regular'. Namely, that solutions remain within the positive quadrant $\mathbb{R}_+^2:=\{ (x,y)\in\mathbb{R}^2:x>0,y>0 \}$ and are bounded in probability, for all $t \geq 0$. The proofs of the theorems in this section are deferred to the Appendix.

\begin{Thm}\label{thm00}
Suppose all parameters in \cref{eq_StochGause} are positive and that the functional response is continuous, nonnegative and satisfies that  $\sup_{x>0} \varphi(x)/x \leq 1$.  Then, for any initial value $(x(0),y(0))\in\mathbb{R}_{+}^{2}$, \cref{eq_StochGause} has, with probability one, a unique solution  $(x(t),y(t))\in\mathbb{R}_{+}^{2}$ for all $t \geq 0$. 
\end{Thm}
\begin{proof}
The main tool for analysis is the Lyapunov operator associated to \cref{eq_StochGause} \citep{mao2007stochastic}
\begin{equation}\label{operator}
L=\frac{\partial}{\partial t}+\left[ x\left( 1-x\right) -\varphi(x)y\right]\frac{\partial}{\partial x}+\left[ \beta \varphi(x)y-\delta y\right]\frac{\partial}{\partial y}+\frac{1}{2}\sigma_x^2x^2\frac{\partial^2}{\partial x^2}+\frac{1}{2}\sigma_y^2y^2\frac{\partial^2}{\partial y^2}
\end{equation}
Following standard techniques, it is enough to show that there is a function $V:\mathbb{R}_+^{2}\rightarrow\mathbb{R}_+$ such that $L[V](x,y)\leq M$ for some positive constant $M$ and for all $x$,$y$ in $\R^2_+$. See \citet{khasminskii2012stability, mao2007stochastic}.

Since the drift and diffusion coefficients in \cref{eq_StochGause} are locally Lipschitz functions, then for any initial value $(x(0),y(0))\in\mathbb{R}_{+}^{2}$ there is an unique solution $(x(t),y(t))$ with $t\in [0,\tau_e)$, where $\tau_e$ is an explosion time. We must to show that $\mathbb{P}(\tau_e=\infty)=1$. Indeed, let $n_0$ be large enough such that $1/n_0 \leq x(0),y(0) \leq n_0$. For each $n\geq n_0$ define the stopping times 
$$\tau_{n}=\inf_{t\in[0,\tau_e]}\left\lbrace x(t)\not\in(1/n,n) \mbox{ or } y(t)\not\in (1/n,n)\right\rbrace,$$		
with $\inf \emptyset :=\infty$. Clearly, $\tau_n$ is an increasing random time and we can define $\tau_{\infty}=\lim_{n \to{+}\infty}{\tau_n} \leq\tau_{e}$ almost surely. Note that for $0\leq t \leq \tau_\infty$, the solution $(x(t),y(t))$ remains in $\R^2_+$.

We proceed by contradiction. If $\mathbb{P}(\tau_{\infty}<\infty)>0$, then there are constants 
$T>0$ and $\epsilon\in(0,1)$ such that $\mathbb{P}\left( \tau_{\infty}\leq T\right) >\epsilon.$
Therefore, there is $n_1\geq n_0$ such that	
$\mathbb{P}\left( \tau_n\leq T\right) \geq\epsilon$ for all $n \geq n_1$. Define $V:\mathbb{R}_{+}^{2}\rightarrow\mathbb{R}_{+}$ by
$$V(x,y)=ax-1-\log(ax)+by-1-\log(by),$$
where $a$ and $b$ are positive constants to be determined. Applying the operator in \cref{operator} to $V$ gives 
\begin{equation}\label{eqA_L}
    L[V]= ax-ax^2-a\varphi(x)y-1+x+\dfrac{\varphi(x)}{x}y+b\beta\varphi(x)y-by-\beta \varphi(x)+1+\dfrac{\sigma_{x}^{2}}{2}+\dfrac{\sigma_{y}^{2}}{2}.
\end{equation}

It\^o's formula and the fact that $\varphi(x) \leq x$ for all $x$ allows us to write
\begin{align*}
\ud V&=L[V] \ud t+\sigma_{x}\left( ax-1\right)x \ud B_{x}(t)+\sigma_{y}\left( ay-1\right)y \ud B_{y}(t)\\
&\leq\left( (a+1)x-\dfrac{ax^2}{k}+(b\beta-a)\varphi(x)y+(1-b)y+\dfrac{\sigma_{x}^{2}}{2}+\dfrac{\sigma_{y}^{2}}{2}\right)\ud t\\
&\quad + \sigma_{x}\left( ax-1\right)x \ud B_{x}(t)+\sigma_{y}\left( ay-1\right)y \ud B_{y}(t)
\end{align*}
Further, taking $a$ and $b$ such that $b>1$ and $a/b>\beta$,  guarantees the existence of a constant $M>0$ such that
\begin{equation}\label{eqA_dV}
    \ud V\leq M \ud t+\sigma_{x}\left( ax-1\right)y \ud B_{x}(t)+\sigma_{y}\left( ay-1\right)y \ud B_{y}(t).
\end{equation}

Integrating both sides of \cref{eqA_dV} between $0$ and $\tau_n\wedge T$, and taking expected value, we obtain
\begin{equation}\label{cota1}
    \EXP ~ V(x(\tau_{n}\land T),y(\tau_{n}\land T))\leq V(x_{0},y_{0})+MT.
\end{equation}
Finally, since $V$ is a non-negative function, we can bound
\begin{eqnarray}\label{cota2}
    \EXP ~ V(x(\tau_{n}\land T),y(\tau_{n}\land T))&\geq & \EXP ~ V(x(\tau_{n}\land T),y(\tau_{n}\land T)I(\tau_n\leq T))\nonumber\\
    & = & \EXP ~ V(x(\tau_{n}),y(\tau_{n}))\nonumber\\
    & \geq & \min\left\lbrace (an-1-\log(an)),(a/n-1-\log(a/n))\right\rbrace
\end{eqnarray}
which contradicts (\ref{cota1}). Therefore $\mathbb{P}(\tau_{\infty}<\infty) =1$ and the solution exists and remains in $\R^2_+$ for all $t \geq 0$.
\end{proof}

Another important feature of the stochastic Gause model is that it has solutions that are bounded in probability. Borrowing the techniques in \citet{GEI199423}, and comparing our model to a mutualistic system, we  establish boundedness in the following general way.

\begin{Thm}\label{thm:bounded}
Let $\theta=(\theta_1,\theta_2)$ be a vector of positive numbers such that $\theta_1+\theta_2<\frac{1}{2}$. Then, under the hypotheses of \cref{thm00}, the solution $(x(t),y(t))$ to the stochastic Gause model in \cref{eq_StochGause} satisfies
\begin{equation}\label{Mao2002}
    \log \left( \mathbb{E}[x(t)^{\theta_1}y(t)^{\theta_2}] \right)\leq e^{-c_1t}(\theta_1\log x(0)+\theta_2\log y(0))+\frac{c_2}{c_1}\left(1-e^{-c_1t}\right)
\end{equation}
for all $t\geq0$, where the positive constants $c_1,c_2$ are given by
\begin{align*}
 c_1 &= \frac{1}{4}(1-\theta_1-\theta_2)\min\{\theta_1\sigma_x^2,\theta_2\sigma_y^2\},\\   
 c_2 &= |\theta| \sqrt{1+\delta^2}+\frac{|\theta|^2 (2+\beta^2 +\sqrt{4+\beta^4})}{8c_1}.
\end{align*}
\end{Thm}
\begin{proof}
Consider the stochastic system 
\begin{equation}\label{mutualism}
    \ud\tilde{x} = \tilde{x}(1 - \tilde{x}) \, \ud t + \sigma_x \tilde{x} \, \ud B_x, \quad
    \ud\tilde{y} = (\beta \tilde{x} \tilde{y} - \delta \tilde{y}) \, \ud t + \sigma_y \tilde{y} \, \ud B_y\nonumber
\end{equation}
The model in \cref{mutualism} represents a mutualistic relationship were species $\tilde{y}$ benefits from the presence of $\tilde{x}$ without affecting $\tilde{x}$. This model is studied in \citet{mao2002environmental}, and from theorem 3.1 of that work, \cref{Mao2002} follows for the process $(\tilde{x},\tilde{y})$.

By the comparison theorem, if $x(0) \leq \tilde{x}(0)$ and $y(0) \leq \tilde{y}(0)$ then, $x(t)\leq \tilde{x}(t)$ and $y(t)\leq \tilde{y}(t) $ almost surely for all $t\geq 0$. Since $\theta_1,\theta_2>0$ and $\theta_1+\theta_2 < \tfrac{1}{2}$ then $x(t)^{\theta_1}y(t)^{\theta_2}\leq \tilde{x}(t)^{\theta_1}\tilde{y}(t)^{\theta_2}$ and \cref{Mao2002} follows for the process $(x,y)$.
    
\end{proof}

Note that by letting $t\to\infty$ in (\ref{Mao2002}), we obtain the  asymptotic mean bound 
\begin{equation}\label{cotaE}
    \limsup_{t\to\infty}\mathbb{E}(x(t)^{\theta_1}y(t)^{\theta_2})\leq e^{c_2/c_1}.
\end{equation}
Moreover, due to Markov inequality for all $\epsilon>0$, there exists $M>0$ such that 
\begin{equation}
 \limsup_{t\to \infty}\mathbb{P}\left( x(t)^{\theta_1}y(t)^{\theta_2}> M \right)\leq \epsilon.   
\end{equation}
In \cref{fig_example} we present examples of solution paths to \cref{eq_StochGause} under various functional responses and compare them with the solutions to the deterministic system.

 \begin{figure}
    \centering
   \begin{subfigure}{7cm}
   \includegraphics[width=\linewidth]{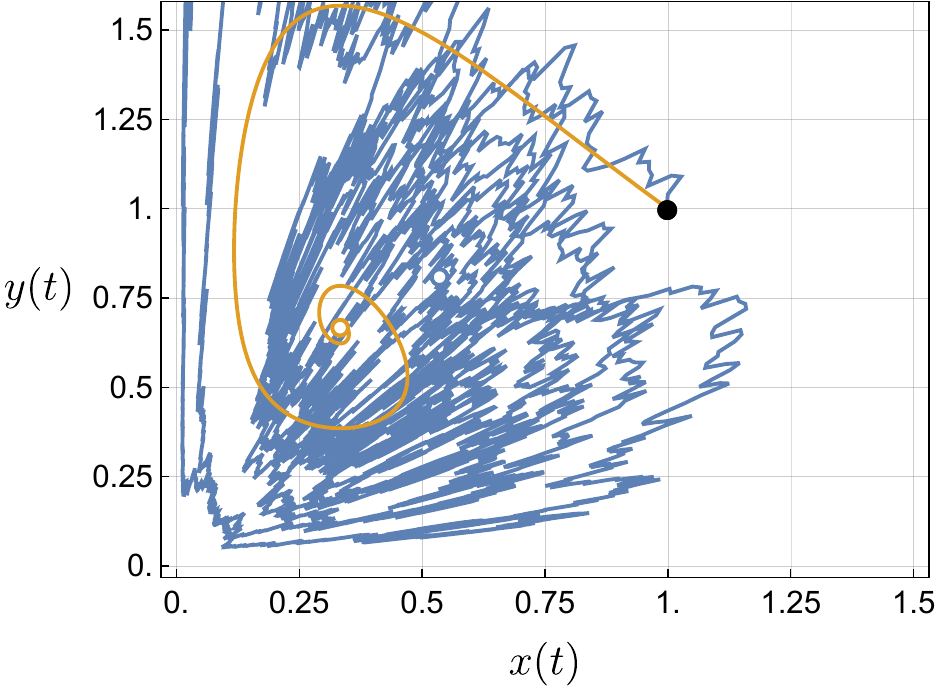}
   \caption{Type I.}
   \end{subfigure}~~
   \begin{subfigure}{7cm}
   \includegraphics[width=\linewidth]{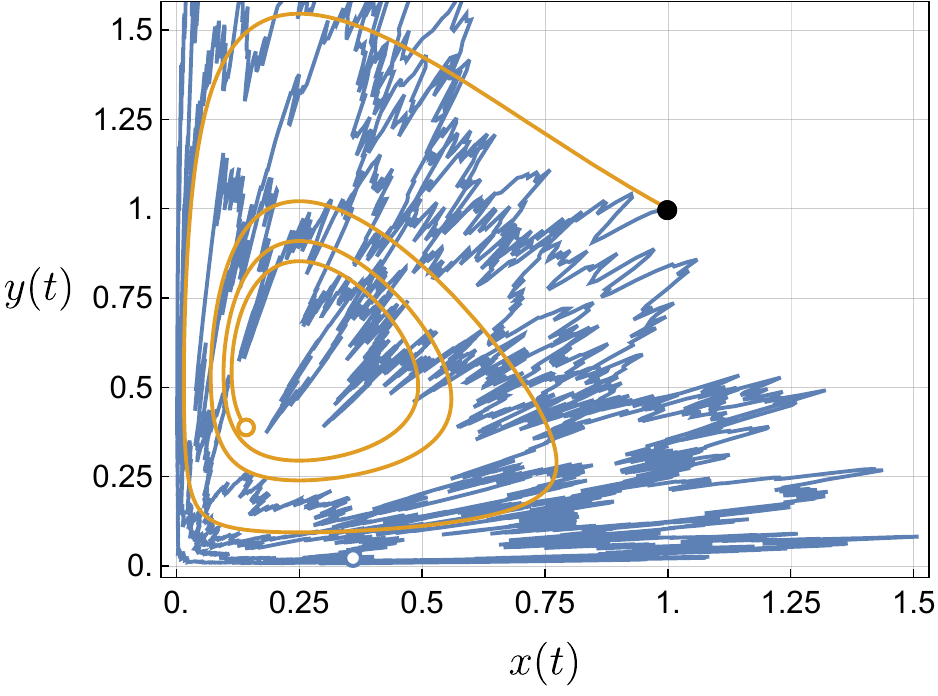}
   \caption{Type II, $\alpha = 0.5$.}
   \end{subfigure}\\[10pt]
   \begin{subfigure}{7cm}
   \includegraphics[width=\linewidth]{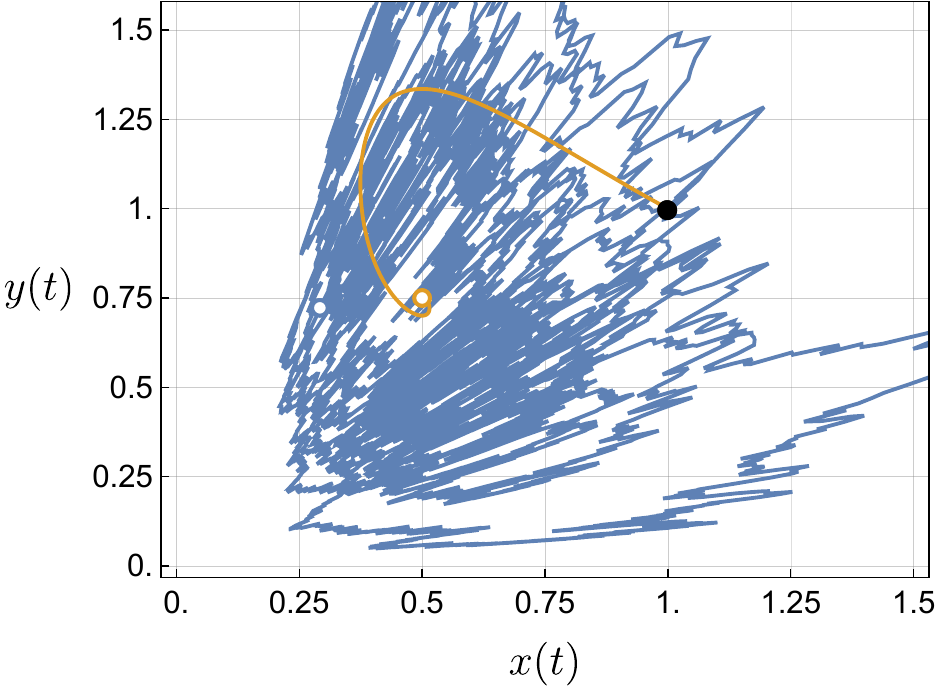}
   \caption{Type III, $\alpha = 0.5$.}
   \end{subfigure}~~
   \begin{subfigure}{7cm}
   \includegraphics[width=\linewidth]{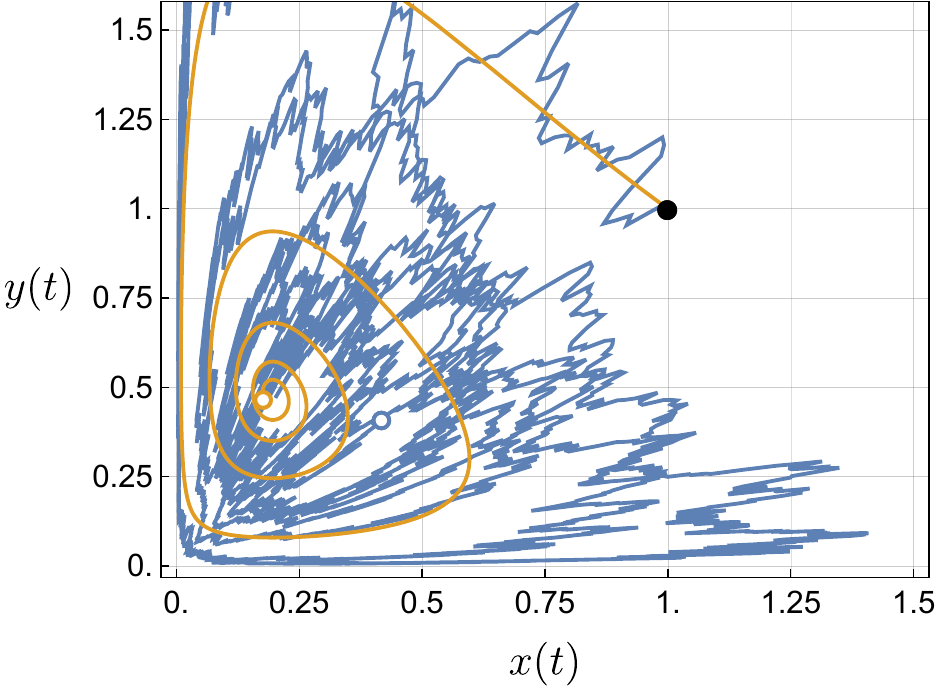}
   \caption{Type IV, $\alpha = 0.25$, $\gamma = 1.5$.}
   \end{subfigure}
   \caption{Examples comparing the solution to the deterministic Gause model in \cref{eq_DetGause} with a realization of the solution to the stochastic counterpart in \cref{eq_StochGause} for each functional response in \cref{tab_FRadim}. Every path has $t \in [0,100]$, $(x(0),y(0)) = (1,1)$, $\beta = 1.5$, $\delta = 0.5$, $\sigma_x^2=\sigma_y^2 = 0.1$. Blank markers show the value of the paths at the final time. In every case, $\beta > \beta^*$ guaranteeing coexistence in the deterministic model.}
   \label{fig_example}
\end{figure}

\section{Noise-induced extinctions}\label{sec:extinction}

We now turn our attention to extinction events, namely, conditions under which the population of one or both species approaches zero. Note from the conditions in \eqref{cond_Determ} that if $(x(0),y(0)) \in \R^2_+$, no extinction can occur in finite time in the deterministic model \eqref{eq_DetGause}. It follows from \cref{thm00} that finite-time extinction is also precluded in the stochastic Gause's model. We thus focus on the events of \textit{eventual extinction}, namely when the population of one or both species approaches zero as $t \to \infty$. We denote these eventual extinction events as
\begin{equation}
    \mathcal{E}_x = \left[ \lim_{t\to\infty} x(t)=0\right], \quad \mathcal{E}_y = \left[ \lim_{t\to\infty} y(t)=0\right].
\end{equation}

It follows from \eqref{cond_Determ} that in the deterministic model, the only possibility for eventual extinction corresponds to the case where $(1,0)$ is asymptotically stable, and the predator $y(t) \to 0$ as $t\to \infty$. This occurs when the mortality rate $\delta$ is higher than $\beta /\varphi(1)$. The prey $x$ in the deterministic model always persists. In contrast, we will show that the presence of noise in the stochastic model can drive either species to eventual extinction. First, we show that as is natural to expect, if the prey becomes extinct eventually, so will the predators.  

\begin{Prop}\label{propA}
Let $(x(t),y(t))$, $t\geq 0$ be the solution to model \cref{eq_StochGause} under the hypotheses of \cref{thm00}. Then $\P(\mathcal{E}_y | \mathcal{E}_x) =  1$.
\end{Prop}

\begin{proof} 

Let $\epsilon\in(0, \frac{\delta}{\beta})$ be fixed. There exists a strictly positive random variable $T$, finite on $\mathcal{E}_x$, such that $0<x(t)<\epsilon$ for all $t\geq T$ on $\mathcal{E}_x$. Since $\varphi(x) \leq x$ for all $x\geq 0$, then the equation for $\ud y$ in \cref{eq_StochGause} can be bounded on $\mathcal{E}_x$ as
\begin{equation}
    \ud y\leq (\beta\epsilon-\delta)y \ud t+\sigma_y y \ud B_y(t)
\end{equation}
for all $t \geq T$. Further, it follows from the positivity of $y(t)$ that 
\begin{equation}
    \ln(y(t))-\ln(y(T)) \leq (\beta\epsilon-\delta)t+\sigma_y[B_y(t)-B_y(T)]
\end{equation}
with probability one conditionally on $\mathcal{E}_x$. Finally, as a consequence of the Law of the Iterated Logarithm and the fact $(\beta\epsilon-\delta)<0$, we obtain that 
$\limsup_{t\to \infty} \frac{1}{t} \log(y(t)) <0$ and therefore $\mathcal{E}_y$ has probability one conditionally on $\mathcal{E}_x$.
\end{proof}

Eventual extinctions can be caused by a sufficiently large noise, even if the values of the model parameters allow for coexistence under deterministic dynamics. We consider first the effect that large noise has on the prey $x$. The next result shows that a sufficiently large value of the variability $\sigma_x$ on the prey dynamics causes its eventual extinction, regardless of the value or the parameters governing the predator. This phenomenon is shared by the one-dimensional logistic model equation  and extends to this case because the predator's role in the dynamics of the $x$ is to increase mortality. See \citet{oksendal2013stochastic, mao2002environmental}. The proof of \cref{thmExt_x} makes this connection explicit.   

\begin{Thm}\label{thmExt_x}
Let $(x(t),y(t))$, $t\geq 0$ be the solution to model \cref{eq_StochGause} under the hypotheses of \cref{thm00}. If $\sigma_{x}^{2}>2$, then $\P(\mathcal{E}_x) = 1$.
\end{Thm}

\begin{proof} 
Let $\hat{x}$ be the solution to $\ud \hat{x} = \hat{x} \ud t + \sigma_x \hat{x} \ud B_x$. Namely,
$$ \hat{x}(t) = \hat{x}(0) \exp\left( \left(1- \frac{\sigma_x^2}{2} \right)t +\sigma_x B_x(t) \right), \quad t\geq 0.$$
If $\sigma_{x}^{2}>2$, the Logarithm Iterated Theorem implies $\hat{x}(t) \to 0$ as $t \to \infty$ almost surely. Since $\varphi(x) \leq x$, the comparison theorem \cite{GEI199423} implies that $x(t) \leq \hat{x}(t)$ for all $t\geq 0$ with probability one. Therefore $x(t) \to 0$ as $t \to \infty$ almost surely. 
\end{proof}

It follows from \cref{propA} and \cref{thmExt_x} that if $\sigma^2_x > 2$, both the predator and prey are driven to eventual extinction with probability one. Namely, the solution $(0,0)$ to \cref{eq_StochGause} is \emph{stochastically assympotically stable}. See \citet{mao2007stochastic}. Sufficiently large noise on the prey can therefore make certain the scenario of complete eventual extinction, which is  always an impossibility in the deterministic model.

For functional responses of types II-IV, where $\varphi$ is bounded, the next result shows that noise alone can cause the extinction of the predator.

\begin{Thm}\label{thmExt_y} 
Suppose all parameters in \cref{eq_StochGause} are positive and that the functional response satisfies $\varphi(x) \leq 1$ for all $x \geq 0$.  Let $(x(t),y(t))$, $t\geq 0$ be the solution to model \cref{eq_StochGause} for $(x(0),y(0)) \in \R^2_+$. If  $\sigma_{y}^{2}>2(\beta-\delta)$, then $\P(\mathcal{E}_y) = 1$.
\end{Thm}

The proof of \cref{thmExt_y} follows exactly the same arguments as that of \cref{thmExt_x}, using the comparison theorem with respect to $\ud \hat{y} = (\beta - \delta) \hat{y} \ud t + \sigma_y \hat{y} \ud B_y(t)$. Note that, since $\sigma_y$ is assumed positive, \cref{thmExt_y} implies that if $\beta<\delta$, any amount of noise will drive the predator $y$ to eventual extinction when the functional response in bounded. This is not unexpected because under those conditions, $\beta < \delta$ implies $\beta < \beta^*$ in \eqref{cond_Determ}, which ensures predator extinction in the deterministic model.

\begin{figure}
    \centering
   \begin{subfigure}{7cm}
   \includegraphics[width=\linewidth]{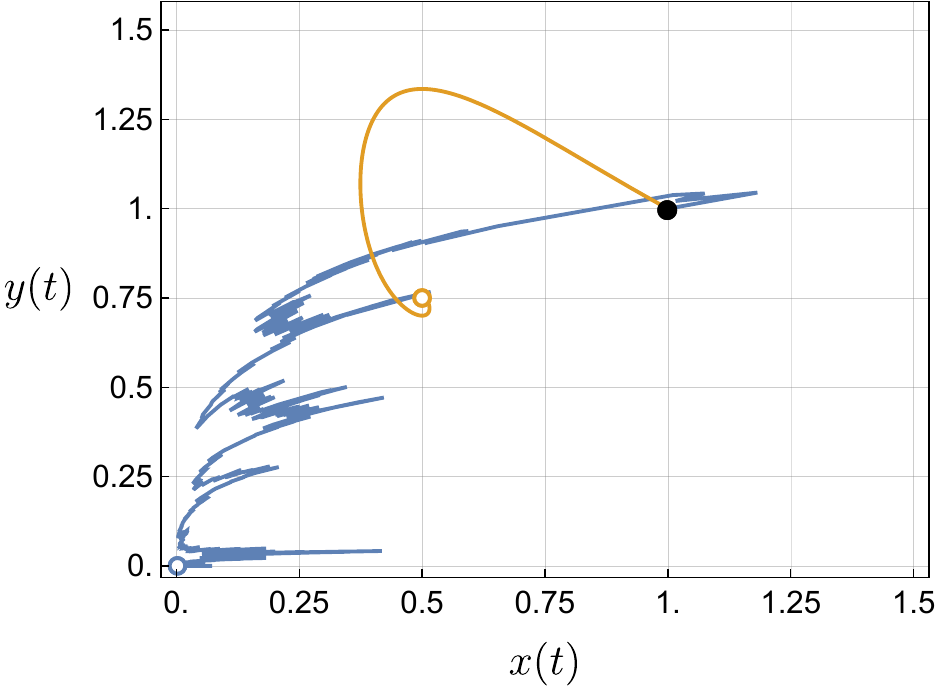}
   \caption{$\sigma_x^2 = 2.1$, $\sigma_y^2 = 0.1$.}
   \end{subfigure}~~
   \begin{subfigure}{7cm}
   \includegraphics[width=\linewidth]{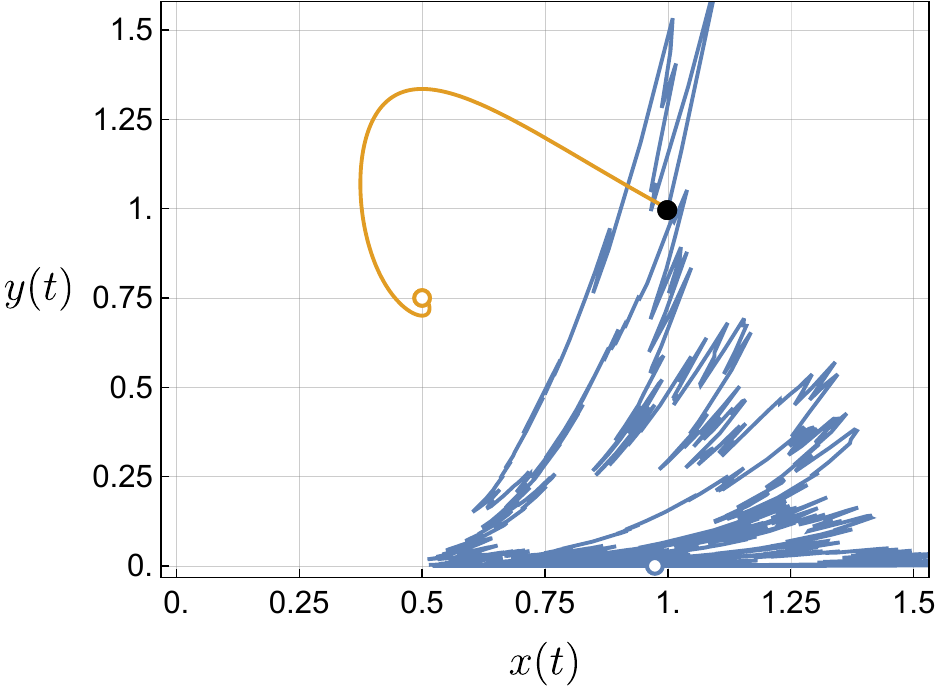}
   \caption{$\sigma_x^2 = 0.1$, $\sigma_y^2 = 2.1$.}\label{fig_Ey}
   \end{subfigure}
   \caption{Examples illustrating \cref{propA}, \cref{thmExt_x} and \cref{thmExt_y} for solutions with Functional Type II. Every path has $t \in [0,100]$, $(x(0),y(0)) = (1,1)$, $\beta = 1.5$, $\delta = 0.5$, $\alpha = 0.5$. In (a) the conditions of \cref{thmExt_x} for $\P(\mathcal{E}_x) = 1$ are satisfied, and hence by \cref{propA}, $y$ also becomes extinct. In (b) only the conditions for the eventual extinction of the predator $y$ are satisfied.}
   \label{fig_ExEy}
\end{figure}

As illustrated in \cref{fig_Ey}, the stochastic Gause model allows for situations where the predator $y$ eventually becomes extinct but the prey $x$ persists. We now turn our attention to the long-term behavior of the  prey in those cases. We prove, in fact, that as $t \to \infty$ and $y(t) \to 0$, the probability distribution of the prey $x(t)$ can converge to a distribution supported on $(0,\infty)$. Moreover, this limiting distribution equals precisely the invariant distribution of the one-dimensional stochastic logistic model.

\begin{Thm}\label{thmY0Xmean}
Let $(x(t), y(t))$, $t \geq 0$ be the solution to the model \cref{eq_StochGause} under the hypotheses of \cref{thm00} and assume $\sigma_{x}^{2} < 2$. Then, for any $y(0) \in (0, \infty)$, and conditional on $\mathcal{E}_y$, the prey process $x(t)$ converges in distribution to a random variable $\tilde{X} > 0$ with probability density function given by 
\begin{equation}\label{eq_distXtilde}
    f_{\tilde{X}}(x) \propto \frac{1}{ \sigma_x^2 x^2} \exp\left(\frac{2(1-x+\log(x))}{\sigma_x^2}\right), \quad x>0,
\end{equation}
\end{Thm}

\begin{proof} 
Let $\epsilon\in(0,1-\frac{1}{2}\sigma_x^2)$ and consider the  following systems of stochastic differential equations:
\begin{align}
    \ud \tilde{x} = (\tilde{x}(1-\tilde{x}))\ud t + \sigma_x \tilde{x} \, \ud B_x(t), &\quad 
    \ud\tilde{y} = (\beta \tilde{x} \tilde{y} - \delta \tilde{y})\ud t + \sigma_y \, \tilde{y} \ud B_y(t), \label{eq_xtilde}\\
    \ud\hat{x} = (\hat{x}(1-\hat{x}) - \varphi(\hat{x})\epsilon)\ud t + \sigma_x \hat{x} \, \ud B_x(t), &\quad
    \ud\hat{y} = (\beta \hat{x} \hat{y} - \delta \hat{y})\ud t + \sigma_y \, \hat{y} \ud B_y(t), \label{eq_xhat}
\end{align}
 where $B_x$ and $B_y$ are the same Brownian motions in \cref{eq_StochGause}. Note that the equations for $\tilde{x}$ and $\hat{x}$ do not depend on their corresponding of $\tilde{y}$ and $\hat{y}$. In, fact the equation for $\tilde{x}$ is that of the one-dimensional stochastic logistic model.    

Let $\{(\tilde{x}(t),\tilde{y}(t)), t\geq 0\}$ be the solution to \cref{eq_xtilde} with $(\tilde{x}(0),\tilde{y}(0)) = (x(0),y(0))$. Since $\sigma_x^2<2$, it is known that $\tilde{x}(t) \to \tilde{X}$  with probability one as $t \to \infty$, where $\tilde{X}$ has density given by \cref{eq_distXtilde} (see \cite{klebanerintroduction} page 171). Moreover, since $\varphi(x)\geq 0$ for all $x$, the comparison theorem of \cite{GEI199423} implies that $x(t)\leq \tilde{x}(t)$ for all $t\geq 0$. It follows that $\limsup_{t\to\infty}x(t)\leq \tilde{X}$ with probability one. 

To establish the behavior of $\liminf x(t)$, fix $T>0$ and define the event $A_{\epsilon,T}=[\sup_{t\geq T} y(t)\leq \epsilon]$. Note the dependence of \cref{eq_xhat} on $\epsilon$ and let $\{(\hat{x}_{\epsilon,T}(t),\hat{y}_{\epsilon,T}(t)), t\geq T\}$ be its solution starting at $(\tilde{x}_{\epsilon,T}(T),\tilde{y}_{\epsilon,T}(T)) = (x(T),y(T))$. Applying the strong Markov property and  the comparison theorem on $A_{\epsilon,T}$, gives that
$\P(\hat{x}_{\epsilon,T}(t) \leq x(t) \text{ for all } t \geq T |A_{\epsilon,T})=1$. Similarly, since $0 < \sigma_x^2 < 2(1-\epsilon)$ the process $\hat{x}_{\epsilon,T}(t)$ converges almost surely as $t\to \infty$ to a random variable $\hat{X}_\epsilon$ supported on $(0,\infty)$. Furhtermore,
\begin{equation}\label{eq_limsxT}
        \P \left(\hat{X}_\epsilon\leq \liminf_{t\to\infty}x(t)\leq \limsup_{t\to\infty}x(t) \leq \tilde{X} \middle| A_{\epsilon,T} \right) =1.
\end{equation}
Letting $T \to \infty$ in \cref{eq_limsxT} we obtain
\begin{equation}\label{eq_limsx}
        \P\left(\hat{X}_\epsilon\leq \liminf_{t\to\infty}x(t)\leq \limsup_{t\to\infty}x(t) \leq \tilde{X} \middle| \limsup_{t \to \infty} y(t) \leq \epsilon \right) =1.
\end{equation}
As $\epsilon \to 0$, from \textbf{(.?.)} $\hat{X}_\epsilon$ converges in distribution to $\tilde{X}$ and $[\limsup_{t \to \infty} y(t) \leq \epsilon] \downarrow \mathcal{E}_y$. This implies that conditional to $\mathcal{E}_y$, $\lim_{t\to\infty}x(t)$ exists with probability one and has the same distribution as $\tilde{X}$.
\end{proof}

\cref{thmY0Xmean} indicates that under the extinction of the predator, if the prey's noise is small enough to allow for its persistence, the prey will behave precisely like in the situation modeled by the one-dimensional stochastic logistic equation. In particular, for large $t$, $x(t)$ exhibits a noise-induced reduction of its carrying capacity. See \cite{klebaner2012introduction}. Namely, the asymptotic behavior of the prey mean abundance goes from  $x(t) \to 1$ in the deterministic case to 
\begin{equation}\label{eq_meanEy}
    \lim_{t \to \infty} \EXP(x(t) | \mathcal{E}_y) = 1- \frac{\sigma_x^2}{2}
\end{equation}
in the stochastic case.

\begin{figure}
    \centering
   \begin{subfigure}{7cm}
   \includegraphics[width=\linewidth]{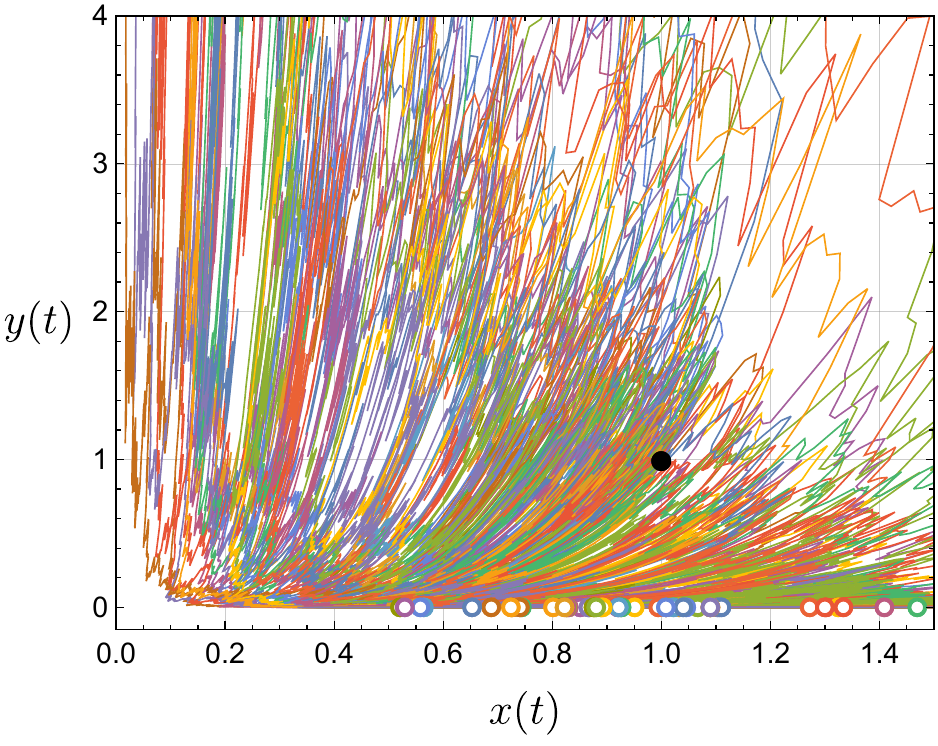}
   \caption{}
   \end{subfigure}~~
   \begin{subfigure}{7cm}
   \includegraphics[width=\linewidth]{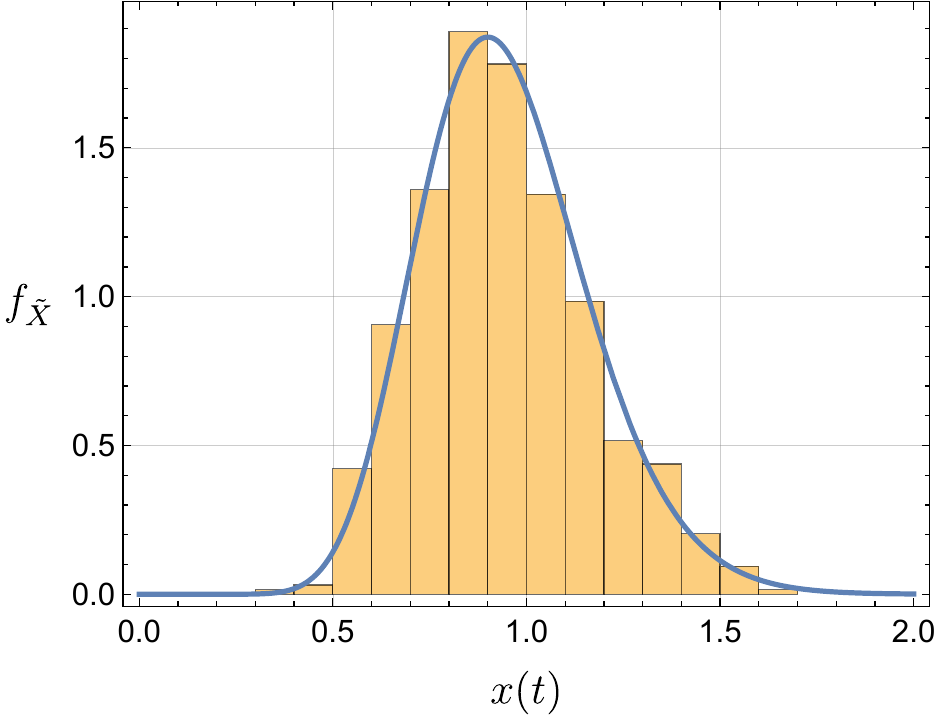}
   \caption{}
   \end{subfigure}
   \caption{Illustration of \cref{thmY0Xmean} for the same problem as \cref{fig_Ey} where $\mathcal{E}_y$ is sure. (a) 50 realizations of the solution for $t \in [0,1000]$. The final values of $x((t),y(t))$ are shown with open markers. (b) Comparison between the sample histogram of the final value of $x(t)$ for 640 realizations, with the density $f_{\tilde{X}}$ in \cref{eq_distXtilde}}.
   \label{fig_EyInv}
\end{figure}

\section{Invariant distribution for type I functional response}\label{sec_invariant}

We now specialize in the case of Type I functional response, namely $\varphi(x) = x$ in \eqref{eq_StochGause}. Note from \eqref{cond_Determ} that for the corresponding deterministic case, the point
\begin{equation}\label{eq_xystar}
    x_* = \frac{\delta}{\beta}, \quad y_* = 1-\frac{\delta}{\beta}
\end{equation}
is a coexistence equilibrium solution that is positive and asymptotically stable if $\delta < \beta$. We will provide conditions on model parameters for the stochastic model \cref{eq_StochGause} to have a unique invariant distribution supported on $\R^2_+$. This corresponds to a situation of stochastic coexistence of prey and predator, and precludes all the eventual extinction scenarios covered in \cref{sec:extinction}. 

\begin{Thm}{\label{thmStatDist}}
    Consider Gause's model \cref{eq_StochGause} with $\varphi(x)=x$, $(x(0),y(0)) \in \R^2_+$ and parameters $\beta,\delta$ such that
    \begin{equation}\label{eq_condbd}
        0<\frac{\beta}{\beta+1} < \delta < \min\left\{\beta,\frac{\beta}{|\beta-1|}\right\}.
    \end{equation}
    If $\sigma_x^2$, $\sigma_y^2$ satisfy that
    \begin{align}
       \frac{\beta ^2 (\beta +4 \delta ) }{\delta  ((\beta -1)
   \delta +\beta )}\sigma _x^2+\frac{(\beta  (\beta +2)-2 \delta ) (\beta
   -\delta ) }{\delta ^2 ((\beta -1) \delta +\beta )}\sigma _y^2&<1,\label{eq_condsxy1}\\
   \frac{\beta ^2 \delta  (\beta +2 \delta ) }{(\beta
   -\delta )^2 (\beta  (1-\delta )+\delta )}\sigma _x^2+\frac{(\beta  (\beta
   +4)-4 \delta ) }{(\beta -\delta ) (\beta  (1-\delta
   )+\delta )}\sigma _y^2&<1 \label{eq_condsxy2},
    \end{align}
    then the model has a unique stationary distribution supported on $\R^2_+$.
\end{Thm}

\begin{proof} 
We establish the existence of an invariant method by the method of Lyapunov functions \cite{mao2011stationary}. Let $V(x,y)=V_{1}(x,y)+V_{2}(x,y)$ with
\begin{align*} 
V_{1}(x,y) &= x-x_{*}-x_{*}\log\left(\dfrac{x}{x_{*}}\right)+\left[y_{*}-y_{*}-\dfrac{1}{\beta}\log \left(\dfrac{y}{y_{*}}\right)\right]\\ 
V_{2}(x,y) &= \dfrac{1}{2}\left[(x-x_{*})+\dfrac{1}{\beta}(y-y_{*})\right]^{2}
\end{align*}
Let $L$ be the Lyapunov operator in \cref{eqA_L}. Then
\begin{align*}
L[V_{1}](x,y) &= -(x-x_{*})^{2}+\dfrac{1}{2}\sigma_{x}^{2} x_{*}+\dfrac{1}{2\beta}\sigma_{y}^{2} y_{*},\\
L[V_{2}](x,y) &= (\sigma_{x}^{2}+y_{*})(x-x_{*})^{2}-\left(x_{*}-\dfrac{y_{*}}{\beta}\right)(x-x_{*})(y-y_{*})\\& \quad -\left(\dfrac{x_{*}}{\beta}-\dfrac{\sigma_{y}^{2}}{\beta^{2}}\right)(y-y_{*})^{2}+\sigma_{x}^{2}x_{*}^{2}+\dfrac{\sigma_{y}^{2}}{\beta^{2}}y_{*}^{2}.
\end{align*}
The conditions for $\beta,\delta, \sigma_x$ and $\sigma_y$ in \cref{eq_condbd,eq_condsxy1,eq_condsxy2} can be equivalently written in terms of conditions on $x_*$ and $y_*$ as
\begin{equation}\label{condxyl}
    x_* - \frac{y_*}{\beta} > 0, \quad \ell(x,y) < \min\{\ell_1(x,y) x_*^2,\ell_2(x,y) y_*^2\},
\end{equation}
where
\begin{align}
    \ell(x,y) &= \sigma_x^2 x_*\left(x_* + \dfrac{1}{2}\right) + \frac{\sigma_y^2}{\beta} \left(\dfrac{y_*}{\beta} + \frac{1}{2}\right), \\
    \ell_1(x,y) &=  1-\sigma_x^2-y_* -\frac{1}{2}\left(x_* - \frac{y_*}{\beta}\right), \\
    \ell_2(x,y) &=  \frac{x_*}{\beta} - \frac{\sigma_y^2}{\beta^2} -\frac{1}{2}\left(x_* - \frac{y_*}{\beta}\right).
\end{align}

We first bound $L[V_2]$. Note that for any $b\geq 0$, $a,c \in \R$,
\begin{equation}
 ax^2-bxy-cy^2<\left(a+\dfrac{b}{2}\right)x^2-\left(c-\dfrac{b}{2}\right)y^2.
\end{equation}
Making $a=\sigma_{x}^{2}+y_{*}$, $c=\dfrac{x_{*}}{\beta}-\dfrac{\sigma_{y}^{2}}{\beta^{2}}$, and $b=x_{*}-\dfrac{y_{*}}{\beta}$, which is positive by \cref{condxyl}, we obtain
\begin{equation*}
    L[V_{2}]\leq \left(\sigma_{x}^{2}+y_{*}-\dfrac{y_{*}}{2\beta}+\dfrac{x_{*}}{2}\right)(x-x_{*})^{2}-\left(\dfrac{x_{*}}{\beta}-\dfrac{\sigma_{y}^{2}}{\beta^{2}}+\dfrac{y_{*}}{2\beta}-\dfrac{x_{*}}{2}\right)(y-y_{*})^{2}+\sigma_{x}^{2}x_{*}^{2}+\dfrac{\sigma_{y}^{2}}{\beta^{2}}y_{*}^{2}.
\end{equation*}
Adding $L[V_1]$ to this bound for $L[V_2]$ gives 
\begin{equation*}
L[V](x,y) \leq -\ell_1(x,y)(x-x_{*})^{2}-\ell_2(x,y)(y-y_{*})^{2}+ \ell(x,y).
\end{equation*}

Note that \cref{condxyl} implies that the ellipsoid $\ell_1(x,y)(x-x_*)^2 + \ell_2(x,y)(y-y_*)^2 \leq \ell$ is entirely contained in $\R^2_+$. We then can choose a neighborhood $U$ of such ellipsoid, and $\epsilon>0$, such that $\bar{U} \in \R^2_+$, and $L[V](x,y) \leq -\epsilon$ for all $(x,y) \in \R^2_+ \setminus U$. Furthermore, the smallest eigenvalue of the diffusion matrix is $\min\{\sigma_x^2 x^2,\sigma_y^2 y^2\}$, which is bounded away from zero in $U$. The existence of the invariant distribution thus follows from Theorem 4.1 in \cite{khasminskii2012stability}.
\end{proof}

Note that the condition in \cref{eq_condbd} for $\beta$ and $\delta$ ensures that the coefficients of $\sigma_x^2$, $\sigma_y^2$ in \cref{eq_condsxy1,eq_condsxy2} are all positive. Therefore, for each admissible pair $(\beta,\delta)$, the set of values $(\sigma_x,\sigma_y)$ for which we can guarantee the existence of an invariant distribution is given by the intersection of two ellipses. This is illustrated in \cref{fig_RegInv}.

\begin{figure}
    \centering
   \begin{subfigure}{7cm}
   \includegraphics[width=\linewidth]{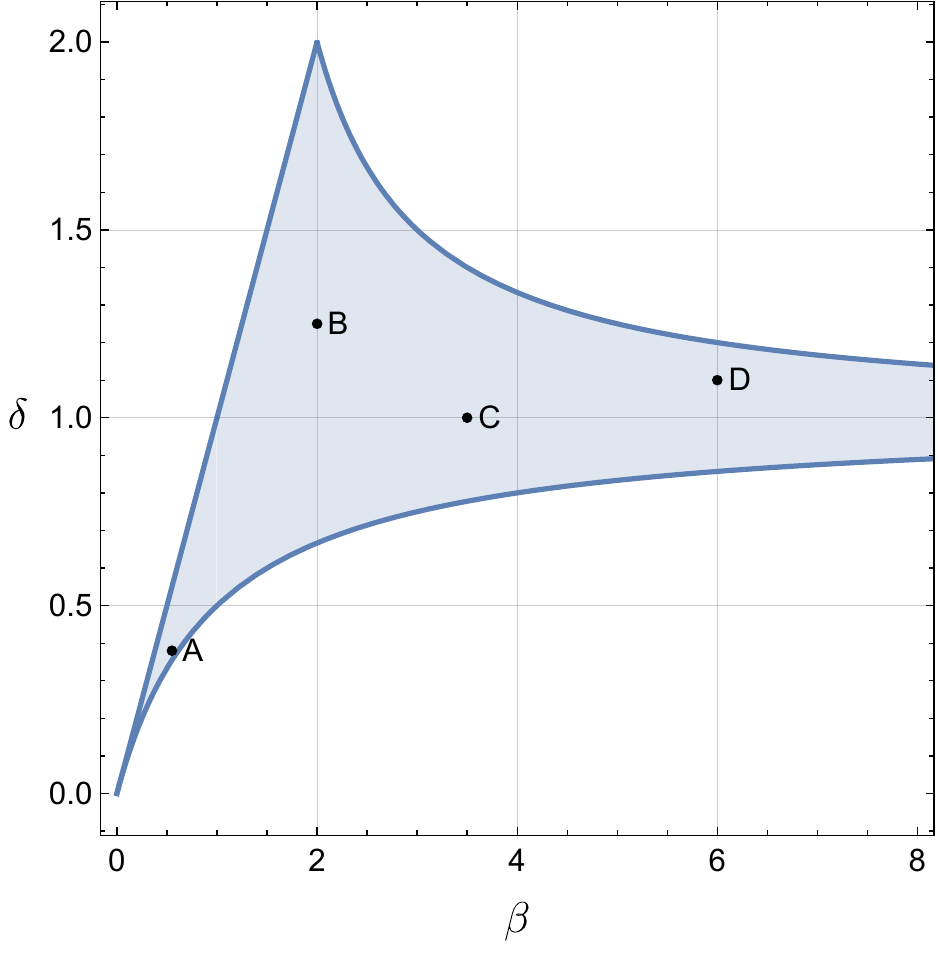}
   \caption{}\label{fig_RegInva}
   \end{subfigure}\hspace{0.25cm}
   \begin{subfigure}{7cm}
   \includegraphics[width=\linewidth]{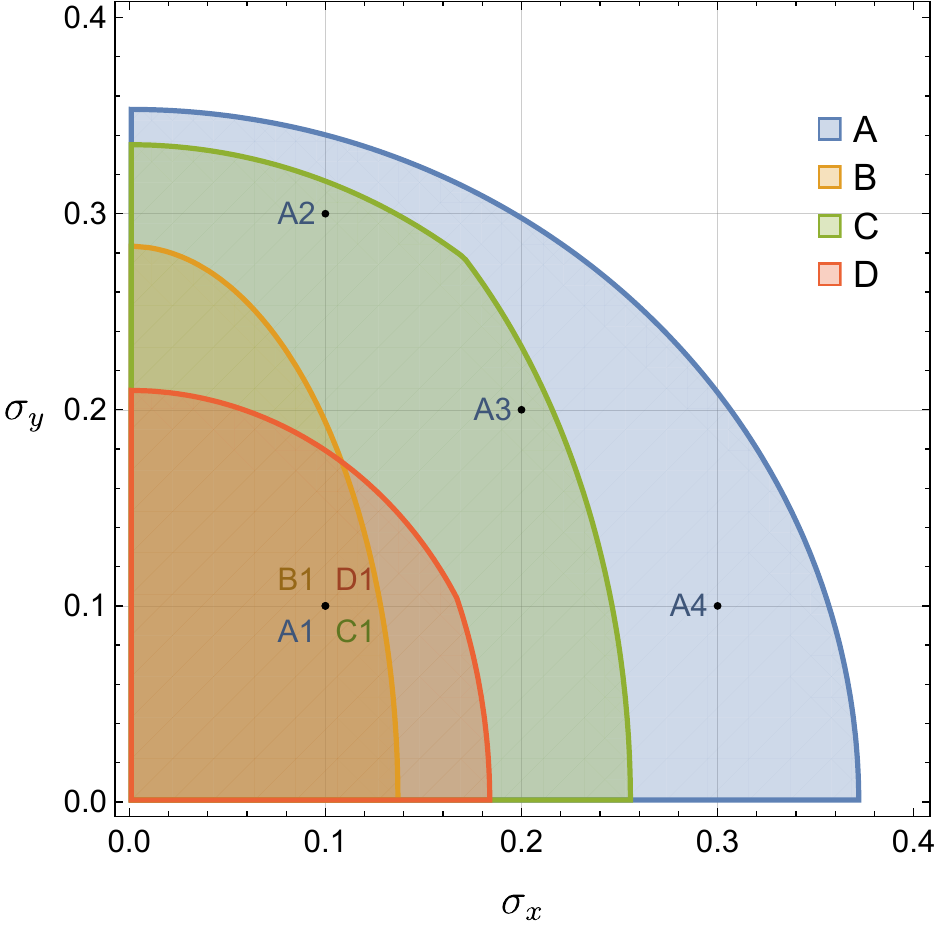}
   \caption{}\label{fig_RegInvb}
   \end{subfigure}
   \caption{Illustration of the conditions of \cref{thmStatDist} for existence of an invariant distribution. (a) Region in the $\beta$-$\delta$ plane determined by \cref{eq_condbd}. The labeled points are: $\mathsf{A} = (0.55,0.38)$, $\mathsf{B} =(2,1.25)$, $\mathsf{C} =(3.5,1)$ and $\mathsf{D} =(6,1.1)$. (b) Regions in the $\sigma_x$-$\sigma_y$ plane determined by \cref{eq_condsxy1,eq_condsxy2} for each of the four points labeled in panel (a).}
   \label{fig_RegInv}
\end{figure}

Under the conditions of \cref{thmStatDist}, the solution $(x(t),y(t))$ to \cref{eq_StochGause} is an ergodic stochastic process. Both the predator and the prey persist with probability one and their populations oscillate around their invariant mean which. See \cref{fig_ErgPaths}. Note that in \cref{fig_ErgPathsb} the conditions for eventual extinction are not met either and both species seem to persist while having highly intermittent paths. \cref{thmStatDist} does not preclude the existence of an invariant distribution in this case, but our numerical computations did not provide evidence of convergence. 

\begin{figure}
   \begin{subfigure}{7cm}
   \includegraphics[scale=0.42]{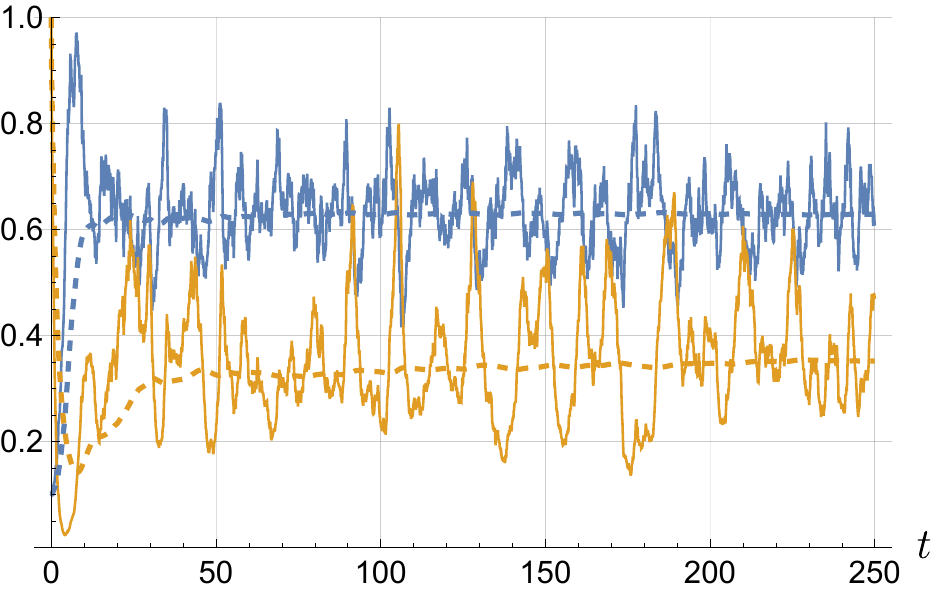}
   \caption{$\beta = 2$, $\delta = 1.25$, $\sigma_x^2=\sigma_y^2 = 0.1$.}
   \end{subfigure}\hspace{-0.2cm}
   \begin{subfigure}{7cm}
   \includegraphics[scale=0.42]{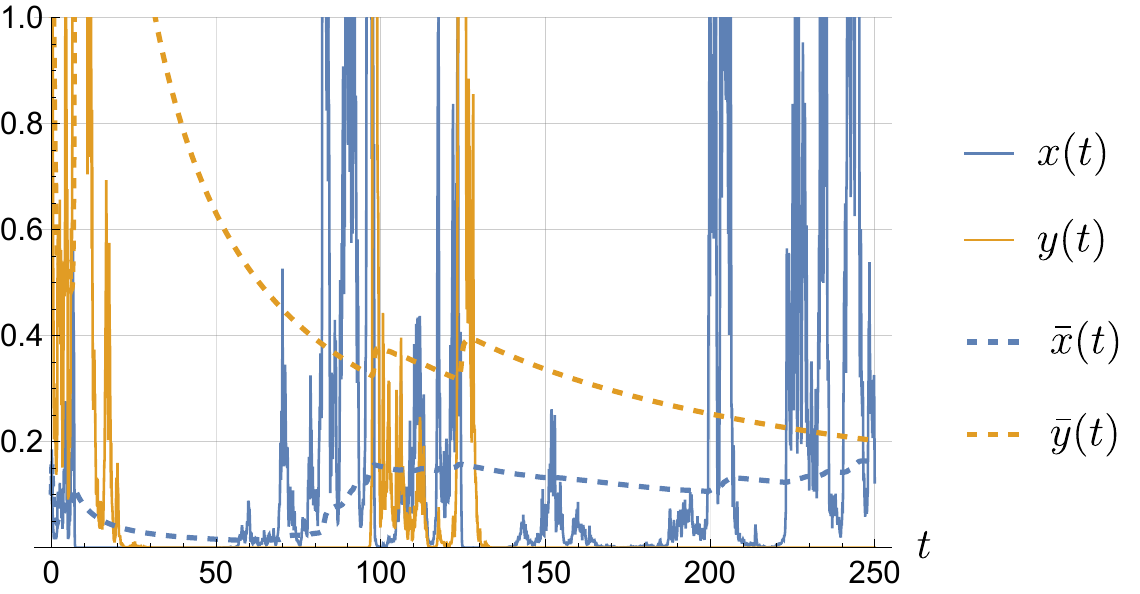}
   \caption{$\beta = 4$, $\delta = 0.1$, $\sigma_x^2=\sigma_y^2 = 1.2$.}\label{fig_ErgPathsb}
   \end{subfigure}
   \caption{Examples of solutions that (a) satisfy the conditions in \cref{thmStatDist} for existence of an invariant distribution, and (b) that do not. Dashed curves denote the Ces\`aro means, namely $\bar{x}(t):= \frac{1}{t} \int_0^t x(s) \ud s$ and similarly for $\bar{y}$, which converge almost surely as $t \to \infty$ whenever a unique invariant distribution exists. In both cases the functional response is of type I, the initial condition is $(x(0),y(0)) = (0.1,1)$ and the solutions are computed for $t \in [0,250]$.}
   \label{fig_ErgPaths}
\end{figure}

The dependence of the invariant distribution on the values of the model parameters is illustrated next. \cref{fig_invHist1} show the dependence on the parameters $\beta,\delta$ for fixed but small values of the noise intensities $\sigma_x, \sigma_y$. Note that the invariant distribution approximately concentrates around the equilibrium solution $(x_*,y_*)$ of the deterministic model, which is stable because $\delta <\beta$. It is apparent that the covariance matrix of the invariant distribution is anisotropic and depends strongly on the values of $\beta$ and $\delta$. \cref{fig_invHistA} shows the sensibility of the invariant distribution to the values of $\sigma_x, \sigma_y$ for a fixed pair of $(\beta, \delta)$. Consistently with the observation in \cref{eq_meanEy}, the difference between the invariant mean and the equilibrium $(x_*,y_*)$ becomes significant for larger values of the noise intensity. 

\begin{figure}
    \centering
    \includegraphics[scale=0.5]{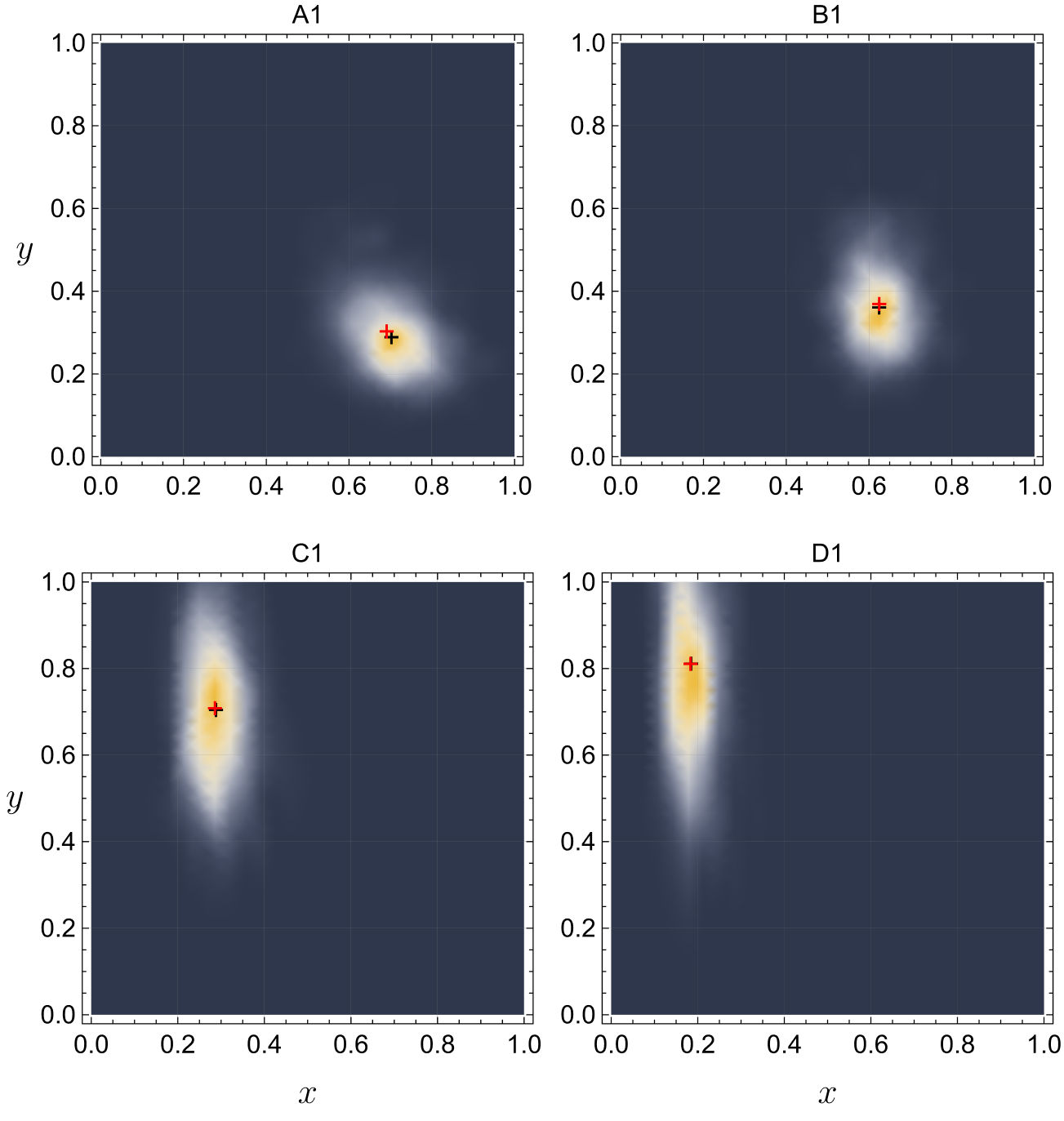}
    \caption{Estimated sample histogram of $(x(t),y(t))$ for $t=2000$ from 1000 realizations with $(x(0),y(0)) = (1,1)$, $\sigma_x^2 = \sigma_y^2 = 0.1$ and four cases of $(\beta,\delta)$ corresponding to the labels in \cref{fig_RegInvb}. The red markers indicate the equilibrium $(x_*,y_*)$ of the deterministic system, and the black shows the estimated mean of $(x(t),y(t))$ at the final time.}
    \label{fig_invHist1}
\end{figure}

\begin{figure}
    \centering
    \includegraphics[scale=0.5]{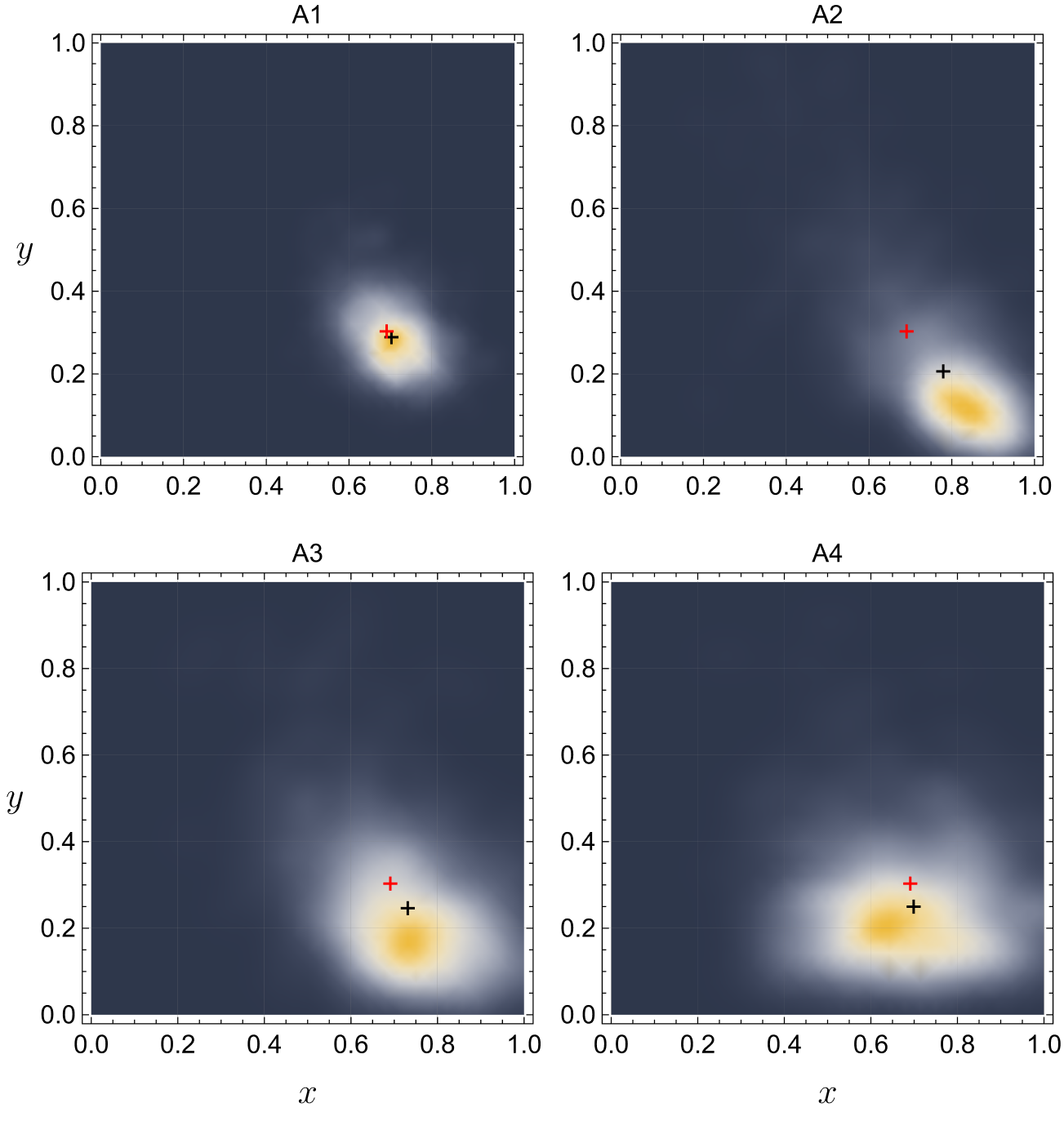}
    \caption{Estimated sample histogram of $(x(t),y(t))$ for $t=2000$ from 1000 realizations with $(x(0),y(0)) = (1,1)$, $(\beta,\delta)=\mathsf{A}$ in \cref{fig_RegInva} and the four different values of $(\sigma_x,\sigma_y)$ as labeled in \cref{fig_RegInvb}. The red markers indicate the equilibrium $(x_*,y_*)$ of the deterministic system, and the black shows the estimated mean of $(x(t),y(t))$ at the final time.}
    \label{fig_invHistA}
\end{figure}

\section{Conclusions and Discussion}

We have conducted a comprehensive analysis of a stochastic version of the Gause population model with linear additive noise. The model considers a prey population \(x\) subject to logistic growth and predation by a predator \(y\), whose functional response can be very general. The specific form of the diffusion term used can be interpreted as modeling the effect of uncertain population parameters with Gaussian errors. In \cref{sec:prel}, we show that the resulting model in \cref{eq_StochGause} is `natural' in the sense that, with probability one, the populations remain bounded and non-negative for all time.

Our results highlight the effect that uncertainty can have on the qualitative properties of predator-prey systems. Perhaps the most intriguing effect is the possibility of noise-induced extinctions, as discussed in \cref{sec:extinction}. We demonstrate in \cref{thmExt_x} and \cref{propA} that sufficiently large noise in the prey's dynamics \(x\) can drive \textit{both} species to eventual extinction, a scenario impossible in the deterministic model. While this result imposes a threshold for \(\sigma_x\) that may not be directly applicable to real ecosystems, it suggests caution when assigning diffusion coefficients in stochastic predator-prey models.

A more nuanced issue is the possibility of eventual extinctions induced by the predator's noise intensity \(\sigma_y\). Our \cref{thmExt_y} indicates that if \(\beta > \delta\), there will always be a level of noise that drives \(x\) to eventual extinction, regardless of the type of coexistence (periodic or otherwise) exhibited in the deterministic case. Moreover, for Type II and III responses, if \(\beta > \beta^* = (1+\alpha) \delta\) and \(\alpha\) is very small, even very small values of \(\sigma_y^2\) are sufficient to drive the predator to extinction. The exclusion of the unbounded functional response of Type I from \cref{thmExt_y} stems from a technical issue when using the comparison theorem with respect to a univariate process. Our numerical simulations lead us to conjecture that the result holds in this case as well, although different proof techniques would be necessary.

\cref{thmY0Xmean} provides a conditional result in the case of an eventual extinction of the predator, whether this extinction is predicted by \cref{thmExt_y} or not. It shows that, in the limit, the prey population behaves as in a single-species system, retaining no information from its interaction with the now-extinct predator. Furthermore, according to \cref{eq_meanEy}, the limiting mean prey population is always lower than the carrying capacity \(y=1\) in the deterministic system, with uncertainty consistently having a detrimental effect. This noise-induced reduction in the limiting mean appears to extend to the case of coexistence, as illustrated in \cref{fig_invHist1,fig_invHistA}.

The conditions provided in \cref{sec:extinction} for noise-induced extinctions are only sufficient, meaning we cannot offer `safe' ranges for noise intensities that would guarantee coexistence. An exception lies in the cases explored in \cref{sec_invariant}, where we present conditions for the existence of an invariant distribution, ensuring coexistence for all time with probability one. We do not claim that our conditions are sharp and note that the critical values for \(\sigma_x\) and \(\sigma_y\) for extinction according to \cref{thmExt_x} and \cref{thmExt_y} are well separated from the region defined by \cref{eq_condsxy1,eq_condsxy2}, where ergodicity is guaranteed. Our numerical simulations suggest that extinctions occur at noise intensities lower than those identified here, and that an ergodic solution is achieved under much less restrictive conditions than those presented in \cref{thmStatDist}. This state of affairs stems from the difficulty of finding the \textit{right} Lyapunov functions that ensure the application of available ergodic theory and leave plenty of room for future research.

\bibliographystyle{plainnat}
\bibliography{main}

\end{document}